\documentclass[preprint]{elsarticle}

\usepackage{subfigure}
\usepackage{epic,eepic}
\usepackage{graphicx,color}

\usepackage{float}
\usepackage{ifthen}
\usepackage{xspace}

\usepackage{amssymb}
\usepackage{amsmath}
\usepackage{amsfonts} 
\usepackage{stmaryrd} 
\usepackage{latexsym} 

\newtheorem{theorem}{Theorem}
\newtheorem{lemma}{Lemma}

\newtheorem{definition}{Definition}
\newtheorem{remarkb}{Remark}
\newtheorem{openq}{Open question}

\newproof{proof}{Proof}

\newcommand\set[1]{\ensuremath{\{ #1 \} }}
\renewcommand\int[1]{\ensuremath{\llbracket#1\rrbracket}}
\newcommand\cur[1]{\ensuremath{{\cal{#1}}}}

\renewcommand{\O}{O}

\begin{document}

\title{Linearity is Strictly More Powerful than Contiguity for Encoding Graphs\tnoteref{t1}\tnoteref{t2}}

\tnotetext[t1]{This work was partially funded by a grant from R{\'e}gion Rh{\^o}ne-Alpes and by the delegation program of CNRS.\\
This work was partially funded by the Vietnam Institute for Advanced Study in Mathematics (VIASM) and by the Vietnam National Fondation for Science and Technology Developement (NAFOSTED).\\
This work was partially funded by Fondecyt Postdoctoral grant 3140527 and N{\'u}cleo Milenio Informaci{\'o}n y Coordinaci{\'o}n en Redes (ACGO).}
\tnotetext[t2]{This is a complete version of the extended abstract appeared in WADS 2015~\cite{CLP+15}.}

\author[ucbl]{Christophe Crespelle}
\ead{christophe.crespelle@inria.fr}

\author[ens]{Tien-Nam Le}
\ead{tien-nam.le@ens-lyon.fr}

\author[amu]{Kevin Perrot}
\ead{kevin.perrot@lif.univ-mrs.fr}

\author[vth]{Thi Ha Duong Phan}
\ead{phanhaduong@math.ac.vn}

\address[ucbl]{Universit\'e Claude Bernard Lyon 1 and CNRS, DANTE/INRIA,\\ LIP UMR CNRS 5668, ENS de Lyon, Universit\'e de Lyon and\\ Institute of Mathematics, Vietnam Academy of Science and Technology,\\ 18 Hoang Quoc Viet, Hanoi, Vietnam.}
\address[ens]{ENS de Lyon, Universit{\'e} de Lyon.}
\address[amu]{Universidad de Chile, DIM, CMM UMR CNRS 2807, Santiago, Chile and\\ Aix-Marseille Universit{\'e}, CNRS, LIF UMR 7279, 13288, Marseille, France.}
\address[vth]{Institute of Mathematics, Vietnam Academy of Science and Technology,\\ 18 Hoang Quoc Viet, Hanoi, Vietnam.}



\begin{abstract}
Linearity and contiguity are two parameters devoted to graph encoding. Linearity is a generalisation of contiguity in the sense that every encoding achieving contiguity $k$ induces an encoding achieving linearity $k$, both encoding having size $\Theta(k.n)$, where $n$ is the number of vertices of $G$. In this paper, we prove that linearity is a strictly more powerful encoding than contiguity, {\em i.e.} there exists some graph family such that the linearity is asymptotically negligible in front of the contiguity. We prove this by answering an open question asking for the worst case linearity of a cograph on $n$ vertices: we provide an $O(\log n/\log\log n)$ upper bound which matches the previously known lower bound.
\end{abstract}

\maketitle

\section{Introduction}

One of the most widely used operation in graph algorithms is the \emph{neighbourhood query}: given a vertex $x$ of a graph $G$, one wants to obtain the list of neighbours of $x$ in $G$. The classical data structure that allows to do so is the adjacency lists. It stores a graph $G$ in $\O(n+m)$ space, where $n$ is the number of vertices of $G$ and $m$ its number of edges, and answers a neighbourhood query on any vertex $x$ in $\O(d)$ time, where $d$ is the degree of vertex $x$.
This time complexity is optimal, as long as one wants to produce the list of neighbours of $x$.
On the other hand, in the last decades, huge amounts of data organized in the form of graphs or networks have appeared in many contexts such as genomic, biology, physics, linguistics, computer science, transportation and industry. In the same time, the need, for industrials and academics, to algorithmically treat this data in order to extract relevant information has grown in the same proportions. For these applications dealing with very large graphs, a space complexity of $\O(n+m)$ is often very limiting. Therefore, as pointed out by~\cite{Turan1984}, finding compact representations of a graph providing optimal time neighbourhood queries is a crucial issue in practice. Such representations allow to store the graph entirely in memory while preserving the complexity of algorithms using neighbourhood queries. The conjunction of these two advantages has great impact on the running time of algorithms managing large amount of data.

One possible way to store a graph $G$ in a very compact way and preserve the complexity of neighbourhood queries is to find an order $\sigma$ on the vertices of $G$ such that the neighbourhood of each vertex $x$ of $G$ is an interval in $\sigma$. In this way, one can store the order $\sigma$ on the vertices of $G$ and assign two pointers to each vertex: one toward its first neighbour in $\sigma$ and one toward its last neighbour in $\sigma$. Therefore, one can answer adjacency queries on vertex $x$ simply by listing the vertices appearing in $\sigma$ between its first and last pointer. It must be clear that such an order on the vertices of $G$ does not exist for all graphs $G$. Nevertheless, this idea turns out to be quite efficient in practice and some compression techniques are precisely based on it~\cite{AD09,BV04,BV05,BSV09,MP10}: they try to find orders on the vertices that group the neighbourhoods together, as much as possible.

Then, a natural way to relax the constraints of the problem so that it admits a solution for a larger class of graphs is to allow the neighbourhood of each vertex to be split in at most $k$ intervals in order $\sigma$. The minimum value of $k$ which makes possible to encode the graph $G$ in this way is a parameter called \emph{contiguity} \cite{GGKS1995} and denoted by $cont(G)$. Another natural way of generalization is to use at most $k$ orders $\sigma_1,\ldots,\sigma_k$ on the vertices of $G$ such that the neighbourhood of each vertex is the union of exactly one interval taken in each of the $k$ orders.
This defines a parameter called the \emph{linearity} of $G$ \cite{CrespelleGambette2009}, denoted $lin(G)$.
The additional flexibility offered by linearity (using $k$ orders instead of just $1$) results in a greater power of encoding, in the sense that if a graph $G$ admits an encoding by contiguity $k$, using one linear order $\sigma$ and at most $k$ intervals for each vertex, it is straightforward to obtain an encoding of $G$ by linearity $k$: take $k$ copies of $\sigma$ and assign to each vertex one of its $k$ intervals in each of the $k$ copies of $\sigma$.

As one can expect, this greater power of encoding requires an extra cost: the size of an encoding by linearity $k$, which uses $k$ orders, is greater than the size of an encoding by contiguity $k$, which uses only $1$ order. Nevertheless, very interestingly, the sizes of these two encodings are equivalent up to a multiplicative constant. Indeed, storing an encoding by contiguity $k$ requires to store a linear ordering of the $n$ vertices of $G$, {\em i.e.} a list of $n$ integers, and the bounds of each of the $k$ intervals for each vertex, {\em i.e.} $2kn$ integers, the total size of the encoding being $(2k+1)n$ integers. On the other hand, the linearity encoding also requires to store $2kn$ integers for the bounds of the $k$ intervals of each vertex, but it needs $k$ linear orderings of the vertices instead of just one, that is $kn$ integers. Thus, the total size of an encoding by linearity $k$ is $3kn$ integers instead of $(2k+1)n$ for contiguity $k$ and therefore the two encodings have equivalent sizes.

Then the question naturally arises to know whether there are some graphs for which the linearity is significantly less than the contiguity. More formally, does there exist some graph family for which the linearity is asymptotically negligible in front of the contiguity? Or are these two parameters equivalent up to a multiplicative constant? This is the question we address here. Our results show that linearity is strictly more powerful than contiguity.

\medskip
\noindent\textbf{Related work.}
Only little is known about contiguity and linearity of graphs.
In the context of $0-1$ matrices, \cite{GGKS1995,WLZ2007} studied closed contiguity and showed that deciding whether an arbitrary graph has closed contiguity at most $k$ is NP-complete for any fixed $k\geq 2$.
For arbitrary graphs again, \cite{GavoillePeleg1999} (Corollary~3.4) gave an upper bound on the value of closed contiguity which is $n/4+\O(\sqrt{n\log n})$.
Regarding graphs with bounded contiguity or linearity, only the class of graphs having contiguity $1$, or equivalently linearity $1$, has been characterized, as being the class of proper (or unit) interval graphs~\cite{Roberts1968}.
For interval graphs and permutation graphs, \cite{CrespelleGambette2009} showed that both contiguity and linearity can be up to $\Omega(\log n/\log \log n)$. For cographs, a subclass of permutation graphs, \cite{CG14} showed that the contiguity can even been up to $\Omega(\log n)$ and is always $O(\log n)$, implying that both bounds are tight. The $O(\log n)$ upper bound consequently applies for the linearity (of cographs) as well, but \cite{CG14} only provides an $\Omega(\log n/\log \log n)$ lower bound.
Finally, let us mention for the sake of completeness that \cite{CG13} gave an algorithm that computes a constant ratio approximation of the contiguity of a cograph, as well as a corresponding encoding, in linear time.

\medskip
\noindent\textbf{Our results.}
Our main result (Theorem~\ref{th:lin-cont}) is to exhibit a family of graphs $G_h$, $h\geq 1$, such that the linearity of $G_h$ is asymptotically negligible in front of the contiguity of $G_h$. In order to do so, we prove (Theorem~\ref{th:upper}) that the linearity of a cograph $G$ on $n$ vertices is always $O(\log n/\log \log n)$. It turns out that this bound is tight, as it matches the previously known lower bound on the worst-case linearity of a cograph~\cite{CG14}.

\medskip
\noindent\textbf{Outline of the paper.}
Section~\ref{sec:prel} gives necessary background on the notions used throughout the article. Section~\ref{sec:fact} proves the key technical statement of our work, showing that the linearity of a cograph is dominated by the maximal height of a certain type of tree, called \emph{double factorial tree}, included in its cotree. From there, Section~\ref{sec:num} derives our main results: the tight upper bound on the linearity of cographs and the construction of a subfamily of cographs for which the linearity is negligible in front of the contiguity.

\section{Preliminaries.}\label{sec:prel}

All graphs considered here are finite, undirected, simple and loopless. In the following, $G$ is a graph, $V$ (or $V(G)$) is its vertex set and $E$ (or $E(G)$) is its edge set. We use the notation $G=(V,E)$ and $n$ stands for the cardinality $|V|$ of $V(G)$. 
An edge between vertices $x$ and $y$ will be arbitrarily denoted by $xy$ or $yx$. The (open) neighbourhood of $x$ is denoted by $N(x)$ (or $N_G(x)$) and its closed neighbourhood by $N[x]=N(x)\cup\set{x}$.
The subgraph of $G$ induced by the set of vertices $X \subseteq V$ is denoted by $G[X] = (X,\{xy \in E \mid x,y\in X\})$.

For a rooted tree $T$ and a node $u\in T$, the depth of $u$ in $T$ is the number of edges in the path from the root of $T$ to $u$ (the root has depth $0$). The \emph{height} of $T$, denoted by $h(T)$, is the greatest depth of its leaves. We employ the usual terminology for {\em children}, {\em father}, {\em ancestors} and {\em descendants} of a node $u$ in $T$ (the two later notions including $u$ itself), and denote by $\cur{C}(u)$ the set of children of $u$. The subtree of $T$ rooted at $u$, denoted by $T_u$, is the tree induced by node $u$ and all its descendants in $T$.
A \emph{monotonic path} $C$ of a rooted tree $T$ is a path such that there exists some node $u \in C$ such that all nodes of $C$ are ancestors of $u$. The unique node of $C$ which has no parent in $C$ is called the root of the monotonic path.

In the following, the notion of \emph{minors} of rooted trees is central. This is a special case of minors of graphs (see e.g. \cite{Lovasz2006}), for which we give a simplified definition in the context of rooted trees.
The \emph{contraction of edge} $uv$ in a rooted tree $T$, where $u$ is the parent of $v$, consists in removing $v$ from $T$ and assigning its children (if any) to node $u$.

\begin{definition}[Minor]\label{def:minor}
A rooted tree $T'$ is a \emph{minor} of a rooted tree $T$ if it can be obtained from $T$ by a sequence of edge contractions. 
\end{definition}

\subsection{Linearity of graphs}

There are actually two notions of linearity (as well as for contiguity, see~\cite{CG14} for definitions) depending on whether one uses the open neighbourhood $N(x)$ or closed neighbourhood $N[x]$.

\begin{definition}[$p$-line-model]\label{def:pline}
A \emph{closed $p$-line-model} (resp. open $p$-line-model) of a graph $G=(V,E)$ is a tuple $(\sigma_1,\ldots
,\sigma_p)$ of linear orders on $V$ such that $\forall v\in V, \exists (I_1,\ldots , I_p)$ such that $\forall i\in \int{1,p},$ $I_i \text{ is an interval of } \sigma_i$ and $N[x]=\bigcup_{1\leq i\leq p} I_i$ (resp. $N(x)=\bigcup_{1\leq i\leq p} I_i$).\\
The \emph{closed linearity} (resp. \emph{open linearity}) of $G$, denoted by $lin(G)$ (resp. $lin^o(G)$), is the minimum integer
$p$ such that there exists a closed $p$-line-model (resp. open $p$-line-model) of $G$.
\end{definition}

\begin{remarkb}
In the definition of a $p$-line-model
, the set of vertices of the intervals $I_i$ assigned to a vertex $x$ are not necessarily disjoint. They are only required to cover the neighbourhood of $x$ while being included in it.
\end{remarkb}

In the rest of the paper, we consider only closed linearity and closed contiguity. But, from~\cite{CG14} and from the inequalities below, for both parameters, the closed notion and the open notion are equivalent. Therefore, the bounds we obtain here (which hold up to multiplicative constants) hold indifferently for open notions and closed notions.

\begin{lemma}\label{lem:ineq}
For an arbitrary graph $G$, we have the following inequalities: $lin(G)-1\leq lin^o(G)\leq 2 lin(G)$.
\end{lemma}

\begin{proof}
The first inequality comes from the fact that an open model can always be turned into a closed model having one additional order $\sigma'$ and where each vertex $x$ of $G$ is assigned a singleton interval of $\sigma'$ equal to $\set{x}$.
Conversely, one can transform a closed model into an open model by duplicating every order $\sigma$ of the closed model into two copies $\sigma_{l}$ and $\sigma_{r}$ in the open model. Then, for each vertex $x$, the interval assigned to $x$ in $\sigma_{l}$ is the left part of the interval ({\em i.e.} vertices of the interval which are before $x$) assigned to $x$ in $\sigma$. And the interval assigned to $x$ in $\sigma_{r}$ is the right part of its interval in $\sigma$.
\end{proof}

Finally, we give two basic properties of linearity that we use in the following.

\begin{remarkb}\label{remark:induced}
The linearity of an induced subgraph of a graph $G$ is at most equal to the linearity of $G$ itself.
\end{remarkb}

Indeed, restricting a $p$-line-model of a graph $G$ to a subset $X$ of its vertices results in a $p$-line-model of $G[X]$.

\begin{remarkb}\label{remark:parallel}
The linearity of the disjoint union $G_\cur{F}$ of a (finite) collection $\cur{F}$ of graphs is the maximum of the linearities of the graphs in $\cur{F}$.
\end{remarkb}

This comes from the fact that a model of $G_\cur{F}$ can be built simply by appending the orders used for the models of the graphs in $\cur{F}$.

\subsection{Cographs}

There are several characterizations of the class of \emph{cographs}. They are often defined as the graphs that do not admit the $P_4$ (path on $4$ vertices) as induced subgraph. Equivalently, they are the graphs obtained from a single vertex under the closure of the parallel composition and the series composition.
The parallel composition of two graphs $G_1=(V_1,E_1)$ and $G_2=(V_2,E_2)$ is the disjoint union of $G_1$ and $G_2$, {\em i.e.}, the graph $G_{par}=\big(V_1 \cup V_2, E_1 \cup E_2\big)$.
The series composition of two graphs $G_1$ and $G_2$ is the disjoint union of $G_1$ and $G_2$ plus all possible edges from a vertex of $G_1$ to one of $G_2$, {\em i.e.}, the graph $G_{ser}\big(V_1 \cup V_2, E_1 \cup E_2 \cup \{xy \mid x \in V_1, y \in V_2\}\big)$. These operations can naturally be extended to a finite number of graphs.

This gives a very nice representation of a cograph $G$ by a tree whose leaves are the vertices of the graph and whose internal nodes (non-leaf nodes) are labelled $P$, for parallel, or $S$, for series, corresponding to the operations used in the construction of $G$. It is always possible to find such a labelled tree $T$ representing $G$ such that every internal node has at least two children, no two parallel nodes are adjacent in $T$ and no two series nodes are adjacent. This tree $T$ is unique~\cite{CLS81} and is called the \emph{cotree} of $G$. See the example on Figure~\ref{fig:cograph}.
Note that the subtree $T_u$ rooted at some node $u$ of cotree $T$ also defines a cograph, denoted $G_u$, and then $V(G_u)$ is the set of leaves of $T_u$.
The adjacencies between vertices of a cograph can easily be read on its cotree, in the following way.

\begin{figure}
  \centering
  $
    \vcenter{\hbox{\includegraphics{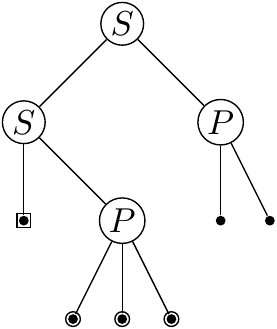}}}
    \hspace{1cm}
    \vcenter{\hbox{\includegraphics{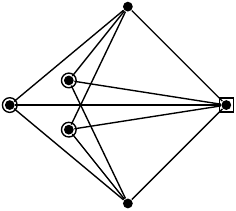}}}
    \hspace{1cm}
    \vcenter{\hbox{\includegraphics{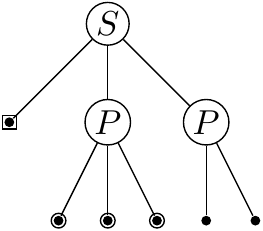}}}
  $
  \caption{Example of tree (left), the cograph it represents (center), and the associated cotree (right). Some vertices are decorated in order to ease the reading.}
  \label{fig:cograph}
\end{figure}

\begin{remarkb}\label{rem:adjcotree}
Two vertices $x$ and $y$ of a cograph $G$ having cotree $T$ are adjacent iff the least common ancestor $u$ of leaves $x$ and $y$ in $T$ is a series node. Otherwise, if $u$ is a parallel node, $x$ and $y$ are not adjacent.
\end{remarkb}

Note that in all the paper, we abusively extend the notion of linearity to cotrees, referring to the linearity of their associated cograph.

\subsection{Comparing power of encodings}

For a graph encoding scheme $Enc$ and a graph $G$, we denote $|Enc(G)|$ the minimum size of an encoding of $G$ based on $Enc$ (there are in general, like here, different encodings based on the same encoding scheme and they do not have necessarily the same size, some being more efficient than others).
We now give a formal definition for an encoding scheme to be strictly more powerful than another one.

\begin{definition}[Strictly more powerful encoding]\label{def:encod}
Let $Enc_1$ and $Enc_2$ be two graph encoding schemes. We say that $Enc_2$ is \emph{at least as powerful as} $Enc_1$ iff there exists $\alpha>0$ such that for all graphs $G$, $|Enc_2(G)|\leq\alpha |Enc_1(G)|$.
Moreover, we say that $Enc_2$ is \emph{strictly more powerful than} $Enc_1$ iff $Enc_2$ is at least as powerful as $Enc_1$ and the converse is not true.
\end{definition}

Note that, $Enc_1$ is not at least as powerful as $Enc_2$ iff there exists a series of graphs $G_h$, $h\geq 1$, such that $|Enc_1(G_h)|/|Enc_2(G_h)|$ tends to infinity when $h$ tends to infinity.
In the introduction, we showed that the encoding schemes $LinEnc$ and $ContEnc$ based on linearity and contiguity respectively are such that, for any graph $G$ on $n$ vertices, we have $2\, n\, cont(G)\leq |ContEnc(G)|\leq 3\, n\, cont(G)$ and $|LinEnc(G)|=3\, n\, lin(G)$. Since $lin(G)\leq cont(G)$, this gives $|LinEnc(G)|\leq \frac{3}{2} |ContEnc(G)|$, showing that linearity is an encoding at least as powerful as contiguity according to Definition~\ref{def:encod}.
In addition, the previous inequalities also imply that $\frac{2}{3} cont(G)/lin(G)\leq |ContEnc(G)|/|LinEnc(G)|\leq cont(G)/lin(G)$. Altogether, we obtain the following remark.

\begin{remarkb}\label{rem:more-power}
Linearity is an encoding at least as powerful as contiguity. Moreover, it is strictly more powerful iff there exists a series of graphs $G_h$, $h\geq 1$, such that $|cont(G_h)|/|lin(G_h)|$ tends to infinity when $h$ tends to infinity.
\end{remarkb}

\section{Linearity of a cograph and factorial rank of its cotree}\label{sec:fact}

In this section, we show that the linearity of a cograph is upper bounded by the size of some maximal structure contained in its cotree, more precisely by the height of a maximal double factorial tree (defined below), which we call the factorial rank of a cotree. This result is interesting by itself as it provides a structural explanation of the difficulty of encoding a cograph by linearity. For our concern, the interesting point is that the number of leaves of a double factorial tree of height $h$ is $\Omega(h!)$. Combined with this fact, the result presented in this section (Lemma~\ref{lem:rank-lin}) will allow us to derive in next section the desired $O(\log n/\log \log n)$ upper bound on the linearity of cographs. We start by some necessary definitions.

\begin{definition}[Double factorial tree]\label{def:fact-tree}
The \emph{double factorial tree} $F^h$ of height $h$ is defined inductively as follows:
\begin{itemize}
\item $F^0$ is the (unique) tree of height $0$, {\em i.e.}, the tree made of one single leaf node, and
\item for $h\geq 1$, $F^h$ is the tree whose root has $2h+1$ children $u$, whose subtrees $F_u$ are precisely $F^{h-1}$.
\end{itemize}
\end{definition}

\begin{definition}[Factorial rank]\label{def:fact-rank}
The \emph{factorial rank} of a rooted tree $T$ (see example on Figure~\ref{fig:factorial}), 
denoted $factrank(T)$, is the maximum height of a double factorial tree being a minor of $T$, that is:\\
$factrank(T) = \max \{h(T')\ |\ T' \text{ is a double factorial tree and a minor of } T\}.$
\end{definition}

\begin{figure}
  \centering
  \includegraphics{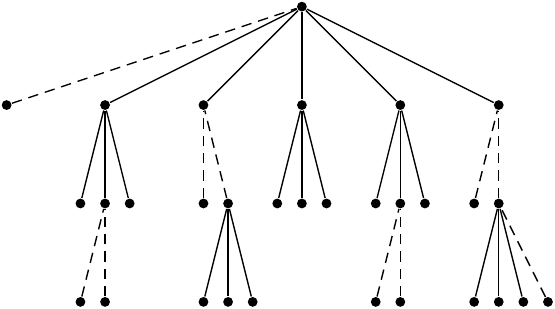}
  \caption{A tree (plain plus dashed arcs) of factorial rank $2$. Plain arcs exhibit a double factorial minor of height $2$, obtained by contracting dashed edges.}
  \label{fig:factorial}
\end{figure}

We extend the notion of factorial rank to a node $u$ in a tree $T$, referring to the factorial rank of its subtree $T_u$.
The case where the children of node $u$ all have factorial rank strictly less than the one of $u$ will play a key role.

\begin{definition}[Minimally of factorial rank $k$]
Let $u$ be a node of a tree $T$. If $u$ has factorial rank $k$ and if all the children of $u$ have factorial rank at most $k-1$, we say that $u$ is \emph{minimally of factorial rank} $k$.
\end{definition}

\begin{figure}
  \centering
  $
    \vcenter{\hbox{\includegraphics{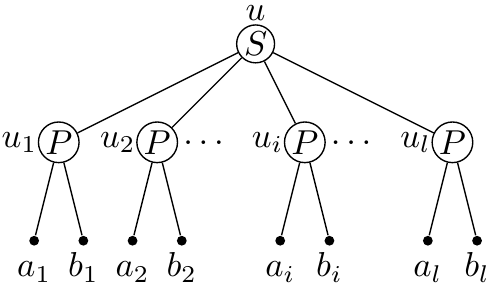}}}
    \hspace{1cm}
    \vcenter{\hbox{\includegraphics{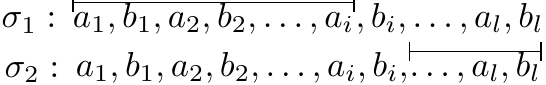}}}
  $
  \caption{Cotree of $G_u$ (left) and example of the intervals for $a_i$ in $\sigma_1$ and $\sigma_2$ (right).}
  \label{fig:initialisation}
\end{figure}

We are now ready to state the result of this section, which claims that the linearity of a cograph is linearly bounded by the factorial rank of its cotree.

\begin{lemma}\label{lem:rank-lin}
Let $T$ be a cotree and let $u \in T$ of factorial rank $k\geq 0$. Then, $lin(G_u)\leq 2k+3$.
Moreover, if $k\geq 1$ and $u$ is minimally of factorial rank $k$, then $lin(G_u)\leq 2k+2$.
\end{lemma}

\begin{proof}
We prove the result by induction. For $k\geq 1$, the induction hypothesis $H(k)$ is formulated as follows: $H(k)=$``all nodes of factorial rank $j\leq k-1$ have linearity at most $2j+3$; and all nodes which are minimally of factorial rank $k$ ({\em i.e.}, whose children have factorial rank at most $k-1$) have linearity at most $2k+2$".

\medskip
\noindent\textbf{Initialisation step.}\\
For the initialisation of our recursion, {\em i.e.} for $k=1$, we must show that if $u$ has factorial rank $0$, then $lin(G_u)\leq 2\times 0+3=3$, and that if $u$ is minimally of factorial rank $1$, then $lin(G_u)\leq 2\times 1 + 2=4$.\\
Firstly, since every internal node of a cotree has at least two children, if $u$ has factorial rank $0$, then $u$ is a leaf of $T$ or $u$ is an internal node having exactly two leaf children (in all other cases, we can find $F^1$ as a minor of $T_u$). Then, it is straightforward that $lin(G_u)=1\leq 3$.\\
Now consider a node $u$ which is minimally of factorial rank $1$, that is $u$ has factorial rank $1$ and all its children have factorial rank at most $0$. If $u$ is a parallel node, then, from Remark~\ref{remark:parallel}, its linearity is the maximum of the linearities of its children, which is $1$ in this case according to what precedes. Thus, we have $lin(G_u)\leq 4$.
If $u$ is a series node, denote its children by $u_1,u_2,\ldots,u_l$. Since all the children of $u$ have factorial rank $0$, as mentioned previously, they are either leaves of $T$ or internal nodes having exactly two leaf children. We consider the case where all of them are internal nodes having two leaf children and we denote $a_i,b_i$ the two leaf children of $u_i$, for $1\leq i\leq l$. We show that in this case, the linearity of $G_u$ is at most $2$ (and so $\leq 4$) by exhibiting a $2$-line-model $(\sigma_1,\sigma_2)$ for $G_u$. As, in the other cases, the graph $G_u$ is an induced subgraph of the graph $G_u$ we consider here, it follows from Remark~\ref{remark:induced} that its linearity is also at most 2 (and so $\leq 4$).\\
Arguments of this paragraph are illustrated on Figure~\ref{fig:initialisation}. For $\sigma_1$ and $\sigma_2$, we use the same order on the vertices of $G_u$, defined as $\sigma_1=\sigma_2=a_1,b_1,a_2,b_2,\ldots,a_l,b_l$. For any $i\in\int{1,l}$, the interval associated to $a_i$ in $\sigma_1$ is the set of vertices less or equal to $a_i$ in $\sigma_1$ and the interval associated to $b_i$ in $\sigma_1$ is the set of vertices greater or equal to $b_i$ in $\sigma_1$. In $\sigma_2$, the interval associated to $a_i$ is the set of vertices strictly greater than $b_i$ in $\sigma_2$ and the interval associated to $b_i$ is the set of vertices strictly less than $a_i$ in $\sigma_2$.\\

\medskip
\noindent\textbf{Induction step.}\\
We consider an integer $k\geq 2$ such that $H(k-1)$ is true, which means in particular that all nodes minimally of rank $k-1$ can be encoded using $2k$ orders. We then show $H(k)$ in two steps: first, we prove that any node $u$ of factorial rank $k-1$ (not necessarily minimally) can be encoded using one more order ({\em i.e.} $2k+1$ orders instead of $2k$ for nodes minimally of rank $k-1$), then we prove that adding again one more order ({\em i.e.} using $2k+2$ orders), we can also encode any node $v$ which is minimally of factorial rank $k$.

\begin{figure}
  \centering
  \includegraphics{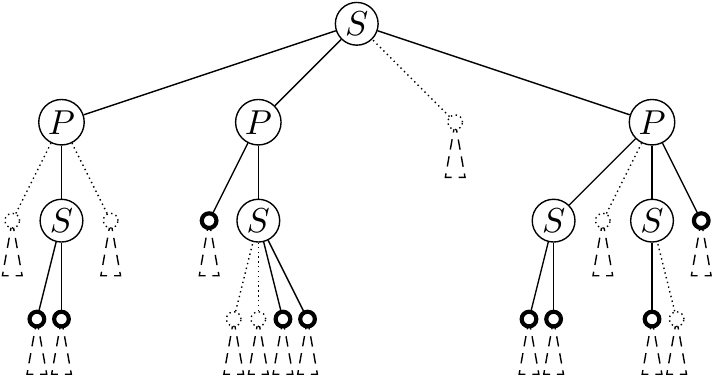}
  \caption{Illustration of the case where $u$ is of factorial rank $k-1$. The top series node is $u$, and the whole cotree is $T_u$. The nodes depicted with non-dotted lines belong to $U_{k-1}$ and the nodes depicted with (non-dotted) bold lines belong to $U_{min}$. 
The tree made of non-dotted nodes and edges is $T'_u$ and dotted nodes belong to $U_{\leq k-2}$. Dashed triangles are remaining parts of $T_u$: subcotrees rooted at nodes in $U_{min} \cup U_{\leq k-2}$.}
  \label{fig:induction-cotree1}
\end{figure}

\medskip
\noindent\textbf{$1^{st}$ step: node $u$ of factorial rank $k-1$.}\\
In order to describe a $(2k+1)$-line-model of $G_u$ we need to distinguish different parts of $T_u$ (see illustration on Figure~\ref{fig:induction-cotree1}).
Let $U_{k-1}$ be the subset of nodes of $T_u$ that have factorial rank $k-1$. If $U_{k-1}$ is reduced to $\set{u}$, then $u$ is minimally of factorial rank $k-1$ and the induction hypothesis allows to conclude without proving anything else. Otherwise, denote $U_{min}=\set{u_1,u_2,\ldots,u_l}\subsetneq U_{k-1}$, where $l\geq 1$, the subset of nodes of $U_{k-1}$ that are minimal for the ancestor relationship ({\em i.e.}, lowest in the cotree). By definition, these elements do not contain node $u$ and are incomparable for the ancestor relationship. Then, one can build a minor of $T_u$, by a sequence of edge contractions, where the set of children of $u$ is exactly $U_{min}$. It follows that $|U_{min}|=l\leq 2k$, as otherwise $u$ would be of factorial rank $k$.
By definition again, all the children of the nodes of $U_{min}$ have factorial rank at most $k-2$, and then the nodes of $U_{min}$ are minimally of rank $k-1$. By induction hypothesis, it follows that for all $i\in\int{1,l}$, $u_i$ admits a $2k$-line-model for which we denote $\sigma_{j}(u_i)$, with $1\leq j\leq 2k$, its $2k$ orders.

We denote $T'_u$ the subtree of $T_u$ induced by the set of nodes $U_{k-1}$ (by definition, $U_{min}\subseteq T'_u$). We also denote $U_{\leq k-2}$ the set of nodes of $T_u\setminus T'_u$ whose parent is in $T'_u\setminus U_{min}$. Nodes of $U_{\leq k-2}$ have, by definition, rank at most $k-2$ and it follows from the induction hypothesis that they admit a $(2k-1)$-line-model. Then, for a node $w\in U_{\leq k-2}$, we again denote $\sigma_j(w)$, with $1\leq j\leq 2k-1$, the $2k-1$ orders of such a model.
In addition , we use an arbitrary partition $\cur{P}$ of the nodes of $T'_u$ into $l$ monotonic paths $C_i$ such that for all $i\in\int{1,l}$, $u_i\in C_i$ (see Figure \ref{fig:induction-cotree2}). Partition $\cur{P}$ naturally induces a generalised partition (some parts may be empty) of $U_{\leq k-2}$ whose parts are denoted $U^i_{\leq k-2}$, with $1\leq i\leq l$: $U^i_{\leq k-2}$ is the subset of nodes of $U_{\leq k-2}$ whose parent belongs to $C_i\setminus\set{u_i}$.

\begin{figure}
  \centering
  \includegraphics{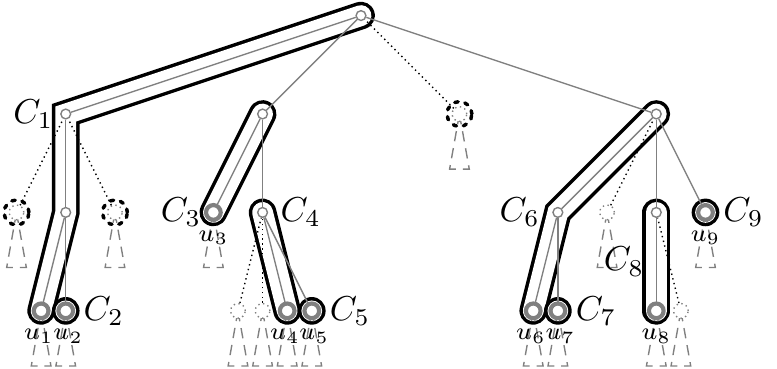}
  \caption{Example of partition into monotonic paths in the case where $u$ is of factorial rank $k-1$, for the cotree of Figure \ref{fig:induction-cotree1}. The three dot circled nodes of $U_{\leq k-2}$ form the set $U^1_{\leq k-2}$.}
  \label{fig:induction-cotree2}
\end{figure}

We can now describe the $2k+1$ orders $(\sigma_j)_{1\leq j\leq 2k+1}$ of the model we build for $G_u$. Importantly, note that $V(G_w)$, $w\in U_{min} \cup U_{\leq k-2}$, is a partition of $V(G_u)$. In our construction, $V(G_w)$ will always be an interval of $\sigma_j$ for all $w\in U_{min} \cup U_{\leq k-2}$ and all $j\in\int{1,2k+1}$. Then, the description of $\sigma_j$ is in two steps: we first give the order, denoted $\pi_j$, in which the intervals of nodes $w\in U_{min} \cup U_{\leq k-2}$ appear in $\sigma_j$ and then, for each $w$, we give the order, denoted $\sigma^w_j$, in which the vertices of $G_w$ appear in this interval. The description of orders $\pi_j$ will be done by choosing a local order on the children of each node of $U_{k-1}\setminus U_{min}$. Then $\pi_j$ is defined as the unique order on $U_{min} \cup U_{\leq k-2}$ respecting all the chosen local orders, {\em i.e.} such that for any $v,v'\in U_{min} \cup U_{\leq k-2}$, if $v$ and $v'$ have the same parent $z$ and if $v$ comes before $v'$ in the order chosen on children of $z$, then all descendants of $v$ come before all descendants of $v'$ in $\pi_j$.

To fully describe the $(2k+1)$-line-model of $u$, we must also assign to each vertex $x$ one interval of its neighbours in each of the orders of the model, in such a way that these intervals entirely cover the neighbourhood of $x$. In order to help our analysis, we distinguish between the \emph{external neighbourhood} of vertex $x$, which is defined as $N[x]\setminus V(G_w)$ (or equivalently $N(x)\setminus V(G_w)$, as $x\in V(G_w)$), where $w$ is the unique node of $U_{min} \cup U_{\leq k-2}$ being an ancestor of leaf $x$ in $T_u$, and its \emph{internal neighbourhood} which is defined as $N[x]\cap V(G_w)$. Our construction starts with the description of the $2k$ first orders of the model, which we use to encode the majority of adjacencies of $G_u$, and finishes with the description of order $\sigma_{2k+1}$ which is used to encode the remaining adjacencies.

For $j\in\int{1,2k}$, the purpose of order $\sigma_j$ is to satisfy the external neighbourhoods of vertices of $G_w$ for $w\in \set{u_j}\cup U^j_{\leq k-2}$. It entirely succeeds to do so for $u_j$ and encodes only one part (out of the two parts that we distinguish in the following) of the external neighbourhoods of $V(G_w)$ for nodes $w\in U^j_{\leq k-2}$, the remaining part being encoded in $\sigma_{2k+1}$. Then, for each $w\in \set{u_j}\cup U^j_{\leq k-2}$, the internal neighbourhoods of vertices of $G_w$ are encoded in the remaining $2k-1$ orders of $(\sigma_j)_{1\leq j\leq 2k}$. It is enough for $w\in U^j_{\leq k-2}$, since they admit a $(2k-1)$-line-model by recursion hypothesis, but one order is missing for $u_j$ which is minimally of rank $k-1$ and is then only guaranteed to admit a $2k$-line-model by recursion hypothesis. Again, the missing order will be found in $\sigma_{2k+1}$.

\medskip
\noindent\textbf{External neighbourhoods and choice of $\pi_j$'s.}
Let $j\in\int{1,2k}$, in this paragraph, we define the order $\pi_j$ in which the intervals of vertices of $G_w$ appear in $\sigma_j$, for $w\in U_{min} \cup U_{\leq k-2}$. If $j>l$, the order $\pi_j$ we choose does not matter, any arbitrary order is suitable. However, if $j\leq l$, the purpose of order $\pi_j$ is to satisfy the external adjacencies of the vertices of $G_{w}$ for any node $w\in \set{u_j}\cup U^j_{\leq k-2}$ (see Figure \ref{fig:induction-cotree2}). In this case, as explained above, we define $\pi_j$ by choosing an order for the children of $u'$ for each node $u'$ of $U_{k-1}\setminus U_{min}$. If $u'$ is an ancestor of $u_j$ and if $u'$ is a parallel node, we choose an order for the children of $u'$ such that the (unique) child of $u'$ which is an ancestor of $u_j$ is the last child in the order (any such order being suitable). If $u'$ is an ancestor of $u_j$ and $u'$ is a series node, we choose an order such that the child of $u'$ which is an ancestor of $u_j$ is the first child of the order (any such order being suitable). And finally, if $u'$ is not an ancestor of $u_j$, then any order on its children is suitable for $\pi_j$. This way, the external neighbourhood of any vertex $x$ of $G_{u_j}$ is exactly the interval of $\sigma_j$ formed by the vertices on the right of the interval of $G_{u_j}$ (containing the last vertex of $\sigma_j$), and this is the interval assigned to $x$ in $\sigma_j$. Indeed, the vertices on the right of the interval of $G_{u_j}$ have a series least common ancestor with node $u_j$ and are therefore adjacent to all the vertices of $G_{u_j}$, while the vertices on the left of the interval of $G_{u_j}$ have a parallel least common ancestor with node $u_j$ and are then non-adjacent to the vertices of $G_{u_j}$ (see example on Figure \ref{fig:induction-cotree3}). As a conclusion of this paragraph, thanks to this choice of $\pi_j$'s, the external neighbourhood of all the vertices of $G_{u_j}$, for all $j\in\int{1,l}$, is entirely encoded in order $\sigma_j$. Also note that the interval associated to the vertices of $G_{u_j}$ in $\sigma_j$, which is the same for all vertices of $G_{u_j}$, contains the last vertex of order $\sigma_j$. We use this property in the $2^{nd}$ step of the induction.

\begin{figure}
  \centering
  \includegraphics{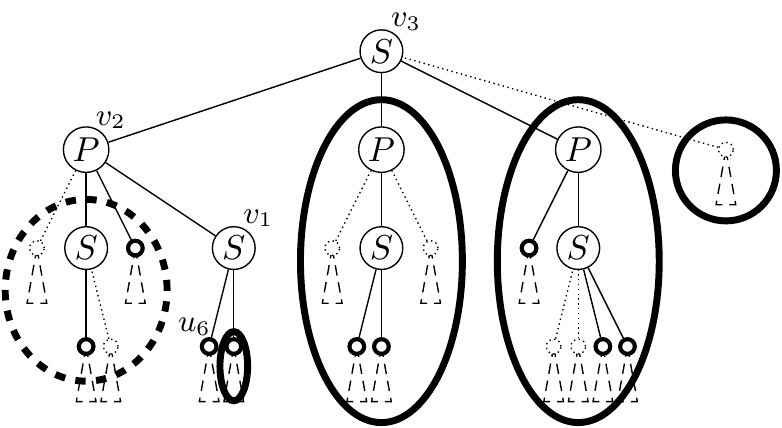}
 \caption{Rearrangement of the cotree of Figure~\ref{fig:induction-cotree1} to obtain order $\pi_6$, aimed at gathering in one interval the external neighbourhood of node $u_6$. For each strict ancestor $v$ of $u_6$, the order on the children of $v$ has been rearranged so that if $v$ is series (nodes $v_1,v_3$ here), the child of $v$ which is an ancestor of $u_6$ is placed first ({\em i.e.} on the left, on the drawing) and if $v$ is parallel (node $v_2$ here), this child is placed last ({\em i.e.} on the right). The order $\pi_6$ is then obtained by reading the nodes of $U_{min} \cup U_{\leq k-2}$ (small circles, bold or dashed, at the top of dashed triangles) from left to right on the drawing: nodes that have a parallel ancestor with $u_6$ (circled by one dashed bold ellipse) appear on the left of $u_6$ and nodes that have a series ancestor with $u_6$ (circled by four bold ellipses) appear on the right of $u_6$.}
  \label{fig:induction-cotree3}
\end{figure}

For a node $w\in U^j_{\leq k-2}$, the situation is slightly more complicated and we consider two cases.
\begin{itemize}
  \item If the father of $w$, denoted $w'$, is a parallel node, then, as previously, the external neighbourhood of vertices of $G_w$ is an interval of $\sigma_j$. Indeed, this external neighbourhood is exactly the set of leaves contained in the subtrees of the children $v$ of the series ancestors of $w$ (which are all strict ancestors of $w'$) such that $v$ is not itself an ancestor of $w$. But, as $w'$ is an ancestor of $u_j$, thanks to the order $\pi_j$ chosen above, this set of leaves is an interval containing the last element of $\sigma_j$. This interval is the one we associate, in $\sigma_j$, to all the vertices of $G_w$.

 \item If the father of $w$, denoted $w'$, is a series node, then the external neighbourhood of vertices of $G_w$ is not an interval of $\sigma_j$ but almost: it is the union of two intervals of $\sigma_j$. Let us distinguish three parts in the external neighbourhood of the vertices of $G_w$. The first part, denoted $A$, is the set of leaves descending from the children $v$ of the series nodes being strict ancestors of $w'$ such that $v$ is not itself an ancestor of $w$. As in the parallel case above, thanks to the choice we made for order $\pi_j$, $A$ is an interval containing the last element of $\sigma_j$. The second part, denoted $B$, is the set of leaves descending from the children of $w'$ that come after $w$ in the order chosen for $\pi_j$. Clearly, $B$ is an interval of $\sigma_j$ and from the definition of $\pi_j$, $B\cup A$ is exactly the interval of vertices, denoted $I_{>w}$, that are on the right of the interval of $G_w$ in $\sigma_j$. This interval $I_{>w}$ is the one we associate to vertices of $G_w$ in $\sigma_j$. Note that it contains the last element of $\sigma_j$. The last part of the external neighbourhood of the vertices of $G_w$ is denoted $I_{<w}$ and is made of the set of leaves descending from the children of $w'$ that precede $w$ in the order chosen for $\pi_j$. As $A$, $I_{<w}$ is an interval of $\sigma_j$, but this part of the external neighbourhood of the vertices of $G_w$ is not covered in $\sigma_j$. This will be done in the additional order $\sigma_{2k+1}$.
\end{itemize}

Before we describe order $\sigma_{2k+1}$, for the purpose of the $2^{nd}$ step of the induction, note that again, the interval of external neighbours associated to any node $w\in U^j_{\leq k-2}$, for any $j\in\int{1,l}$, contains the last vertex of order $\sigma_j$.

We now define the order $\pi_{2k+1}$ used to build order $\sigma_{2k+1}$, using the partition of $T'_u$ into paths $C_i$ introduced earlier. To define $\pi_{2k+1}$, for any node $u'\in U_{k-1}\setminus U_{min}$, we use the same order on the children of $u'$ as the one used for $\pi_i$, with $i\in\int{1,l}$ such that $u'\in C_i$. This ensures that for any node $w\in U_{\leq k-2}$ whose parent $w'$ is a series node of $C_i$, the interval $I_{<w}$ of external neighbours which was not covered in order $\sigma_i$ (note that since $w'\in C_i$ then $w\in U^i_{\leq k-2}$) will also be an interval of $\sigma_{2k+1}$. This is precisely the interval we assign to vertices of $G_w$ in $\sigma_{2k+1}$, which is possible as their internal neighbourhood will be entirely satisfied in the $2k$ first orders, as described below.

\medskip
\noindent\textbf{Internal neighbourhoods and choice of $\sigma^w_j$'s.}
The orders $\sigma^w_j$ used for the vertices of $G_w$, with $w\in U_{min} \cup U_{\leq k-2}$, in order $\sigma_j$, with $j \in \int{1,2k}$, are chosen as follows.

\begin{itemize}
  \item For any node $u_i\in U_{min}$, with $i\in\int{1,l}$, and all $j \in \int{1,2k}$,
  \begin{itemize}
    \item if $j=i$, then we can take any arbitrary order for the vertices of $G_{u_i}$. Indeed, in $\sigma_i$, the vertices of $G_{u_i}$ have already been assigned an interval made only of their external neighbours (see above), meaning that this interval does not contain any vertex of $G_{u_i}$.
    \item If $j\neq i$, the order on the vertices of $G_{u_i}$ is $\sigma_j(u_i)$ and the interval associated to vertices of $G_{u_i}$ in $\sigma_j$ is the same as the one associated to them in $\sigma_j(u_i)$.
  \end{itemize}
  \item For a node $w\in U_{\leq k-2}$ and all $j \in \int{1,2k}$, the order we choose for the vertices of $G_w$ depends on the path $C_i$, with $1\leq i\leq l$, of the partition $\cur{P}$ to which belongs the father of $w$.
  \begin{itemize}
    \item If $j=i$ then we use any arbitrary order for the vertices of $G_w$. Again, in $\sigma_i$, the vertices of $G_{w}$ have already been assigned an interval made only of their external neighbours (see above) and therefore it does not contain any vertex of $G_{w}$.
    \item If $j<i$ (resp. if $j>i$) then we use the order $\sigma_j(w)$ (resp. $\sigma_{j-1}(w)$), and the interval of $\sigma_j$ associated to the vertices of $G_w$ is the same as the one associated to them in $\sigma_j(w)$ (resp. $\sigma_{j-1}(w)$).
  \end{itemize}
\end{itemize}

In this way, for any $i\in\int{1,l}$ and for any $w\in U^i_{\leq k-2}$, since $G_w$ needs only $2k-1$ orders to be encoded, all the internal adjacencies of vertices of $G_w$ have been covered by the intervals associated to them in orders $\sigma_j$, for $1\leq j\leq 2k$ and $j\neq i$. For nodes $u_i$, $1\leq i\leq l$, the situation is the same: only $2k-1$ orders, namely the orders $\sigma_j$ for $1\leq j\leq 2k$ and $j\neq i$, have been used to encode the internal neighbourhoods of $G_{u_i}$. But unfortunately, since $u_i$ is minimally of factorial rank $k-1$, the recursion hypothesis only guarantees that $lin(G_{u_i})\leq 2k$. Then, one more interval is needed to fully cover the internal neighbourhood of vertices of $G_{u_i}$. For this, we use one additional order $\sigma_{2k+1}$.

Actually, we already used order $\sigma_{2k+1}$ in what precedes, in order to cover the external neighbourhood of some vertices. To this purpose, we fixed the order $\pi_{2k+1}$ in which the intervals of vertices of $G_w$, for $w\in U_{min} \cup U_{\leq k-2}$, appear in $\sigma_{2k+1}$. But we still have the liberty of choosing the orders $\sigma^w_j$ on the vertices of $G_w$, for all $w\in U_{min} \cup U_{\leq k-2}$. We use this possibility for each node $u_i\in U_{min}$: we choose the order on the vertices of $G_{u_i}$ in $\sigma_{2k+1}$ as being $\sigma_i(u_i)$, the one which has not been used until now, and the interval associated to vertices of $G_{u_i}$ in $\sigma_{2k+1}$ is the same as the one associated to them in $\sigma_i(u_i)$. This is possible as the external neighbourhood of vertices of $G_{u_i}$ has already been entirely satisfied before, in order $\sigma_i$.

Thus, using the $2k+1$ orders described above, both the internal and the external neighbourhoods of the vertices of $G_w$, for all $w\in U_{min} \cup U_{\leq k-2}$, have been covered. Since $\set{V(G_w)}_{w\in U_{min} \cup U_{\leq k-2}}$ is a partition of the vertices of $G_u$, this proves that $lin(G_u)\leq 2k+1$ and this achieves the $1^{st}$ step of the induction. Also remember, as we use it in the $2^{nd}$ step of the induction described below, that in the model we built for $G_u$, for any vertex $x$ there exists an index $j\in\int{1,2k+1}$ such that the interval associated to $x$ in $\sigma_j$ contains the last vertex of $\sigma_j$.

\medskip
\noindent\textbf{$2^{nd}$ step: node $v$ minimally of factorial rank $k$.}\\
In order to finish the induction step and then the proof of Lemma~\ref{lem:rank-lin}, we now show that for a node $v$ minimally of factorial rank $k$ ({\em i.e.}, whose children have factorial rank at most $k-1$), we have $lin(G_v)\leq 2k+2$.

First consider the case where $v$ is a parallel node. In this case, from Remark~\ref{remark:parallel}, the linearity of $v$ is the maximum of the linearity of its children. Since the children of $v$ all have factorial rank at most $k-1$, it follows from the $1^{st}$ step of our induction that their linearity is at most $2k+1$. Consequently, we have $lin(G_v)\leq 2k+1$, and then in particular $lin(G_v)\leq 2k+2$.

Let us now consider the case where $v$ is a series node and let us denote $v_1,v_2,\ldots,v_l$, with $l\in\mathbb{N}$, the children of $v$. From what precedes, all of them have linearity at most $2k+1$ and for each $i\in\int{1,l}$ we have a $(2k+1)$-line-model of $G_{v_i}$ denoted $(\sigma_j(v_i))_{j \in \int{1,2k+1}}$. A remarkable property of this $(2k+1)$-line-model, which we have constructed above, is that for any vertex $x$ of $G_{v_i}$, there exists an index $j\in\int{1,2k}$ such that the interval associated to $x$ in $\sigma_{j}(v_i)$ contains the last vertex of $\sigma_{j}(v_i)$. 
For each vertex $x$, we denote $ind(x)$ such an index $j$. We now use this property in order to construct a $(2k+2)$-line-model of $G_v$, which we denote $(\sigma_1,\ldots,\sigma_{2k+1},\sigma_{2k+2})$.

For any $j\in\int{1,2k+1}$, the order $\sigma_j$ used for $v$ is simply the concatenation (denoted $+$) of the orders of the $(2k+1)$-line-models of its children, from left to right in increasing value of the index. More explicitly, for all $j\in\int{1,2k+1}$, we define $\sigma_j$ as $\sigma_j=\sigma_j(v_1)+\sigma_j(v_2)+\ldots+\sigma_j(v_l)$. For any $i\in\int{1,l}$ and for any vertex $x$ of $G_{v_i}$, if $j\neq ind(x)$, the interval associated to $x$ in $\sigma_j$ is the same as the one associated to $x$ in $\sigma_j(v_i)$. On the other hand, if $j=ind(x)$, as the interval associated to $x$ in $\sigma_{ind(x)}(v_i)$ contains the last vertex of $\sigma_{ind(x)}(v_i)$, in the order $\sigma_{ind(x)}$ of the model of $G_v$, we extend this interval on the right by including the vertices of $G_{v_{i'}}$ for all $i'>i$. As $v$ is a series node, all these vertices are indeed adjacent to $x$ (see Figure~\ref{fig:model}).

\begin{figure}
  \begin{center}
  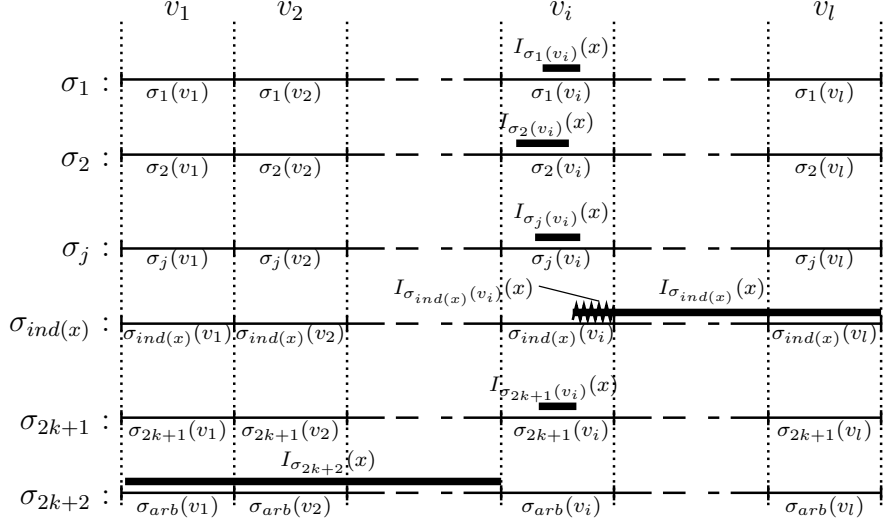
\end{center}
  \caption{Illustration of the $(2k+1)$-line-model of $G_v$ for $v$ minimally of factorial rank $k$, and of the intervals associated to a vertex $x$ of $G_{v_i}$, for some $i\in\int{1,l}$. The interval associated to $x$ in an order $\sigma$ is denoted $I_{\sigma}(x)$.}
  \label{fig:model}
\end{figure}

In this way, for any $i\in\int{1,l}$ and for any vertex $x$ of $G_{v_i}$, the internal neighbourhood of $x$ is entirely covered in the orders $\sigma_1,\ldots,\sigma_{2k+1}$. Regarding the external neighbourhood of $x\in V(G_{v_i})$, note that it can be expressed as $\bigcup_{i'\in\int{1,l} \text{ and } i'\neq i}V(G_{v_{i'}})$. The part $\bigcup_{i'>i}V(G_{v_{i'}})$ is already covered in order $\sigma_{ind(x)}$. Then, only the part $\bigcup_{i'<i}V(G_{v_{i'}})$ of the external neighbourhood of $x$ remains to be covered. This is the purpose of order $\sigma_{2k+2}$ which we define as follows. For $i\in\int{1,l}$, we take any arbitrary order $\sigma_{arb}(i)$ on the vertices of $G_{v_i}$ and we build $\sigma_{2k+2}$ as $\sigma_{2k+2}=\sigma_{arb}(1)+\sigma_{arb}(2)+\ldots+\sigma_{arb}(l)$. Then, for any $i\in\int{1,l}$ and for any vertex $x$ of $G_{v_i}$, we associate to $x$ the interval of $\sigma_{2k+2}$ made of the vertices of $\bigcup_{i'<i}V(G_{v_{i'}})$ (see Figure~\ref{fig:model}). Doing so, the entire external neighbourhood of all the vertices of $G_{v}$ are covered in the $2k+2$ orders we defined. Thus, $(\sigma_1,\ldots,\sigma_{2k+1},\sigma_{2k+2})$ is a $(2k+2)$-line-model of $G_{v}$ which is then of linearity at most $2k+2$.

This completes the induction step and the proof of Lemma~\ref{lem:rank-lin}.
\end{proof}

\section{Main results}\label{sec:num}

The first result we derive from Lemma~\ref{lem:rank-lin} is a tight upper bound on the worst-case linearity of cographs on $n$ vertices. Until now, the best known upper bound \cite{CG14} was $O(\log n)$, and \cite{CG14} also exhibits some cograph families having a linearity up to $\Omega(\log n/\log\log n)$. Here, we show a new upper bound of $\O(\log n/\log\log n)$ that matches the lower bound of \cite{CG14}. This is a direct consequence of Lemma~\ref{lem:rank-lin} and of the fact that a double factorial tree of height $h$ has $\Omega(h!)$ leaves.

\begin{theorem}\label{th:upper}
For any cograph $G$ on $n$ vertices, we have $lin(G)=\O(\log n/\log\log n)$, and this upper bound is tight.
\end{theorem}

\begin{proof}
Let $T$ denote the cotree of $G$ and $k=factrank(T)$. From Lemma \ref{lem:rank-lin}, the linearity of $G$ is in $\O(k)$. Let us now show that $k=\O(\log n/\log\log n)$, which will conclude this proof. According to the definition of factorial rank, $G$ has at least as many vertices as the number of leaves of the double factorial tree of height $k$, which has $\prod_{i=0}^{k} (2i+1)$ leaves. It follows from Stirling's approximation of factorial that
  $$n \geq \prod \limits_{i=0}^{k} (2i+1) = \frac{(2(k+1))!}{2^{k+1}(k+1)!} \geq \frac{2\sqrt{\pi}}{e} \left(\frac{2(k+1)}{e}\right)^{k+1}$$
  and consequently
  $$\log n \geq (k+1) \left( \log(k+1) + \log\left(\frac{2}{e}\right) \right) + \log\left(\frac{2\sqrt{\pi}}{e}\right) \geq (k+1) \big( \log(k+1)-1 \big).$$
  As $x \geq y >1$ implies $\frac{x}{\log x} \geq \frac{y}{\log y}$, we have
  $$\frac{\log n}{\log\log n} \geq \frac{(k+1) \big( \log(k+1)-1 \big)}{\log(k+1)+\log \big( \log(k+1)-1 \big)}$$
  and it follows that $k = \O(\log n / \log\log n)$.\\
And finally, as \cite{CG14} exhibits some cographs having linearity $\Omega(\log n / \log\log n)$, consequently, the upper bound provided by the lemma is tight. 
\end{proof}

We now prove the main result aimed by this paper: linearity is a strictly more powerful encoding than contiguity, which means, from Remark~\ref{rem:more-power}, that there exists some graph families for which the linearity is asymptotically negligible in front of the contiguity (hereafter denoted $cont(G)$ for a graph $G$).

\begin{theorem}\label{th:lin-cont}
There exists a series of graphs $G_h$, $h\geq 1$, such that $cont(G_h)/lin(G_h)$ tends to infinity when $h$ tends to infinity. 
\end{theorem}

\begin{proof}
For $h\geq 1$, let $G_h$ be the connected cograph whose cotree is a complete binary tree of height $h$ and let $n=2^h$ denote the number of vertices of $G_h$. It is proven in \cite{CG14} that $cont(G_h)=\Theta(\log n)$ and that $lin(G_h)=\Omega(\log n/\log\log n)$. Then, Theorem~\ref{th:upper} above implies that $lin(G_h)=\Theta(\log n/\log\log n)$ and therefore $cont(G_h)/lin(G_h)=\Theta(\log\log n)$, which achieves the proof.
\end{proof}

\section{Perspectives}
In this paper, we showed that linearity provides a strictly more powerful encoding for graphs than contiguity does, meaning that the ratio between the contiguity and the linearity of a graph is not bounded by a constant. From a practical point of view, the meaning of our result is that using several orders, instead of just one, for grouping neighbourhoods of vertices can greatly enhance compression rates in some cases.


We obtained our main result by exhibiting a graph family, namely a subfamily of cographs, for which the ratio between the contiguity and the linearity tends to infinity as fast as $\Omega(\log\log n)$, with $n$ the number of vertices in the graph.
As a by-product of our proof, but meaningful in itself, we also showed tight bounds for the worst-case linearity of cographs on $n$ vertices; tight bounds were previously known for contiguity.
Several questions naturally arise from these results and others.

\begin{openq}
What is the worst case contiguity and the worst-case linearity of arbitrary graphs?
\end{openq}

It is straightforward to see that both of these values are bounded by $n/2$. Conversely, since there are $2^{n(n-1)/2}$ graphs on $n$ labelled vertices and since contiguity and linearity do not depend on the labels of the vertices, then both encodings must use at least $n^2$ bits for graphs on $n$ vertices. Moreover, when the value of the parameter is $k$, the size of the corresponding encoding is $O(k\, n)$ integers, that is $O(k\, n \log n)$ bits. Consequently, both parameters must be at least $\Omega(n/\log n)$ in the worst case. For contiguity, \cite{GavoillePeleg1999} gave an upper bound asymptotically equivalent to $n/4$. Is $\Omega(n)$ indeed the worst-case contiguity of a graph? Is the worst-case for linearity the same as the one for contiguity? Another appealing question which is closely related is the following.

\begin{openq}
For arbitrary graphs, what is the maximum gap between contiguity and linearity?
\end{openq}

In other words, let $(G_n)_{n\geq 1}$ be a family of graphs on $n$ vertices and let $f(n)=cont(G_n)/lin(G_n)$. Can $f(n)$ tends to infinity faster than $\Omega(\log\log n)$? What is the maximum asymptotic growth possible for $f(n)$?
Answering those questions would be both theoretically and practically of key interest for the field of graph encoding.

\section*{References}

\bibliographystyle{elsarticle-harv}
\bibliography{cographLinearity}

\end{document}